\documentclass{article}

\newif\ifextended
\extendedtrue

\usepackage[a4paper]{geometry}
\usepackage{amsmath,amssymb}
\usepackage{url}
\usepackage{hyperref}
\usepackage{cleveref}
\usepackage{csquotes}
\usepackage{orcidlink}
\usepackage{tcolorbox}
\usepackage{amsmath}
\usepackage{amsthm}
\usepackage{algorithm}
\usepackage{algpseudocodex}
\usepackage{tikz}
\usepackage{xspace}
\usepackage{leftindex}
\usepackage{ebproof}
\usepackage{varwidth}
\usepackage{titling}

\usepackage{todonotes}
\newtheorem{theorem}{Theorem}
\newtheorem{lemma}{Lemma}

\def\hrulei{\hrule height0pt\relax}

\usetikzlibrary{arrows,fit,positioning}

\AtBeginDocument{%
  }

\makeatletter
\pgfdeclareshape{datastore}{
  \inheritsavedanchors[from=rectangle]
  \inheritanchorborder[from=rectangle]
  \inheritanchor[from=rectangle]{center}
  \inheritanchor[from=rectangle]{base}
  \inheritanchor[from=rectangle]{north}
  \inheritanchor[from=rectangle]{north east}
  \inheritanchor[from=rectangle]{east}
  \inheritanchor[from=rectangle]{south east}
  \inheritanchor[from=rectangle]{south}
  \inheritanchor[from=rectangle]{south west}
  \inheritanchor[from=rectangle]{west}
  \inheritanchor[from=rectangle]{north west}
  \backgroundpath{
    \southwest \pgf@xa=\pgf@x \pgf@ya=\pgf@y
    \northeast \pgf@xb=\pgf@x \pgf@yb=\pgf@y
    \pgfpathmoveto{\pgfpoint{\pgf@xa}{\pgf@ya}}
    \pgfpathlineto{\pgfpoint{\pgf@xb}{\pgf@ya}}
    \pgfpathmoveto{\pgfpoint{\pgf@xa}{\pgf@yb}}
    \pgfpathlineto{\pgfpoint{\pgf@xb}{\pgf@yb}}
 }
}
\makeatother

\newcommand{\inextended}[2]{\ifextended {#1} \else {#2} \fi}

\newcommand{\minusS}[0]{\ensuremath{\mathit{MinusS}}\xspace}

\newcommand{\minusLang}[0]{\ensuremath{\mathit{MinusLang}}\xspace}

\newcommand{\minusRelPrime}[2]{\ensuremath{\leftindex_{#1}{\overset{#2}{\rightsquigarrow}}^{\prime}}\xspace}
\newcommand{\ruleKeyword}{\textbf{rule}\ \xspace}

\begin{document}

\title{Minuska: Towards a Formally Verified Programming Language Framework}

\author{Jan Tušil\orcidlink{0000-0002-7264-2569} \\\href{mailto:jan.tusil@mail.muni.cz}{jan.tusil@mail.muni.cz} \and
Jan Obdržálek\orcidlink{0000-0002-6655-7798} \\\href{mailto:obdrzalek@fi.muni.cz}{obdrzalek@fi.muni.cz}}

\begin{titlingpage}
\maketitle

\begin{abstract}
  Programming language frameworks allow us to generate language tools
  (e.g., interpreters) just from a formal description of the syntax and
  semantics of a programming language. As these frameworks
  tend to be quite complex, an issue arises whether we can trust the generated
  tools.  To address this issue, we introduce a practical formal programming
  language framework called Minuska, which always generates a provably correct
  interpreter given a valid language definition.
  This is achieved by
  (1) defining a language MinusLang for expressing programming language
  definitions and giving it formal semantics and (2) using the Coq proof
  assistant to implement an interpreter parametric in a MinusLang definition
  and to prove it correct.  Minuska provides strong correctness guarantees and
  can support non-trivial languages while performing well.
  This is the extended version of the SEFM'24 paper of the same name.
\end{abstract}

\end{titlingpage}

\newcommand{\K}{\mbox{$\mathbb{K}$}\xspace}

\newcommand{\ghrepo}[0]{https://github.com/h0nzZik/minuska}
\newcommand{\ghtag}[0]{v0.3.0}
\newcommand{\ghref}[2]{\href{\ghrepo/blob/\ghtag/#1}{#2}}

\section{Introduction}\label{sec:intro}
To rigorously reason about software that controls computer systems -- e.g., to
guarantee their security, reliability, or safety -- one must be able to assign
precise meaning to programs. This can be rigorously done by providing
\emph{formal semantics} to programming languages. It has been argued,
e.g., in~\cite{KleinCDEFFMRTF12}, that encoding a language's
semantics \emph{in a mechanized formal language} can yield both theoretical
and practical benefits.  Using a practical \emph{semantic framework} such as
PLT-Redex~\cite{KleinCDEFFMRTF12} or \K~\cite{RosuS10}, one can take a formal
description of a programming language and derive from it a parser, an
interpreter~\cite{KleinCDEFFMRTF12,HillsSR07}, or a deductive verifier~\cite{StefanescuPYLR16}.

In another line of research, recent years have witnessed the emergence of
formally verified programming language tools.  For example,
CompCert~\cite{BoldoJLM13} is a compiler for the C programming language,
formally verified in Coq, and CakeML~\cite{KumarMNO14} is a compiler for a
variant of ML, verified using HOL4.

In~\cite{ChenLTR20}, the authors have presented, within the auspices of the
\K semantics framework, their vision of combining the
best of the two worlds: A semantics framework which, given some formal syntax and
semantics of a programming language,
can produce the tools above together with guarantees of their correctness.
These guarantees are mathematical proofs, which can be mechanically
machine-checked by a trustworthy proof checker.

The paper~\cite{ChenLTR20} also contains the first step to achieving this
vision.
The authors formalize program execution as mathematical proof
and create a tool that generates a \emph{certifying interpreter} (given a formal description of a programming language).
This interpreter is given a program and its
input and produces the result of running the program
on the input and a \emph{certificate} in the form of a complete execution trace,
more precisely, its associated proof object.
A simple Metamath~\cite{metamath}-based proof checker can subsequently check this proof object.
This way, one can
detect if the interpreter does not behave correctly (the execution trace does
not conform to the formal semantics of the language) on a given program
input.

In this work, we advance further.
Instead of having an interpreter that
would provide us with machine-checkable certificate (proof) that the execution was correct, we
want to produce an interpreter which is \emph{provably correct}. This way we
not only get stronger guarantees, but also do not have to spend
computational resources on checking the generated certificates.  Thus, there
are at least two motivations for having such a verified semantics framework:
\emph{correctness}, and \emph{performance}.

\subsection{Our approach}

In our work, we have built a new semantics framework, called
Minuska, instead of working within the \K framework. There are
several reasons for this. First, while being
human-readable, the \K input language
is not suitable for formal verification since it does
not have formal semantics. Therefore, \cite{ChenLTR20}~works with the
intermediate Kore language,
to which the \K framework translates its input language.
While Kore is much simpler to handle, it does not have formal
semantics either -- tools from the \K framework handle its various
features (such as subsorting and attributes) in an ad-hoc way.
In~\cite{ChenLTR20}, the authors, therefore, target only its strict subset,
which is then given formal semantics by translating it into a logic named
``matching logic''~\cite{MatchingMuLogic}. While this fragment of Kore is
expressive enough to encode simple rewriting systems like the one used as a
running example in~\cite{ChenLTR20} or the ones included therein for
performance evaluation, it
is not expressive enough to fully handle even the simple IMP programming
language definition, which serves as a motivating example of the paper:
according to~\cite{ChenLTR20}, their system supports neither the use of
built-in collection datatypes nor evaluation order attributes, both of which
IMP relies on.

A second reason we have decided to go with creating a new framework (and a new
input language) is that the fragment of Kore supported in~\cite{ChenLTR20}
allows for language definitions that are internally inconsistent -- due to the
presence of function symbols. Indeed, one can have a user-defined function
symbol $s$ and axiomatize it to be equal to the integer value $3$ and the
integer value $4$ simultaneously, and the \K framework will
happily accept such a definition. With such a contradiction, we could generate
``proofs'' which do not guarantee anything.

We designed the input language of Minuska to avoid the issues above. It
has a simple but formal semantics, and its interpreter is written in the Coq
proof assistant~\cite{Coq} and formally proven correct.  
The input language of Minuska is
strong enough to express the motivating example of~\cite{ChenLTR20}, which the
\K-based certifying interpreter of~\cite{ChenLTR20} could not handle, but our verified
interpreter can.
Moreover, our formally verified interpreter outperforms the certifying interpreter of~\cite{ChenLTR20},
which was said to show \blockquote{promising performance},
by a factor of at least 20.

\subsection{Contributions}

Our main technical contributions are:
\begin{enumerate}
    \item \emph{MinusLang} -- a language embedded in Coq, with formal semantics, for specifying operational semantics
    of programming languages, with no room to introduce logical inconsistency to the language definition;
    \item an \emph{interpreter} for any programming language that has a well-formed MinusLang semantics;
    \item a mechanized \emph{proof of the soundness and completeness} of that interpreter.
\end{enumerate}

On the conceptual level, the main takeaways of our paper are the following:
\begin{enumerate}
    \item One can have a programming language framework that is practical, formally verified, and inexpensive to develop.
    (It took about 4 person-months to develop the framework with the verified interpreter).
    \item The use of foundational formal verification can result in a significantly more performant tool
    than retrofitting after-the-task certificate generation and checking to a professionally developed but unverified tool,
    even without clever optimizations on the side of the verified tool.
\end{enumerate}

The rest of the paper is organized as follows: We start by giving a short
overview of semantic frameworks in~Section~\ref{sec:background}. The following
presents Minuska -- in Section~\ref{sec:trs} we present the term rewriting
system on which MinusLang is based and Section~\ref{sec:builtins} is devoted to
the treatment of built-in types (which are one of the important aspects
of Minuska). The universal interpreter is described in
Section~\ref{sec:interpreter} and finally Section~\ref{sec:trust} discusses
the usage of Minuska and its inherent trust base. We follow by examing the
performance of Minuska, both with repect to~\cite{ChenLTR20} and on its own, in
Section~\ref{sec:evaluation}, which also discusses Minuska's limitations and
future work. {\textbf{This is the extended version, with appendix, of the SEFM'24 paper of the same name.}} The formal syntax and semantics of MinusLang, as well as the full
definition of semantics for a simple imperative language, can be found in the \inextended{appendix.}{extended version~\cite{MinuskaExt} of this paper.}
Minuska is available at:
\\
\begin{centering}
\href{\ghrepo/tree/\ghtag}{\ghrepo/tree/\ghtag}
\end{centering}
and on Zenodo~\cite{MinuskaArtifact}.


\section{Semantic frameworks and their correctness}\label{sec:background}

A \emph{programming language semantic framework} is a software tool in which
formal semantics of various programming languages can be encoded and that
derives other artefacts from the semantics.  Among the most widely known
semantic frameworks, there is Ott~\cite{SewellNOPRSS10} that generates
parsers, \LaTeX{} code, and theorem prover (Coq, HOL, and Isabelle/HOL)
definitions from a definition of a language's syntax and semantics;
\K~\cite{ChenR18} that generates parsers, interpreters, and deductive
verifiers; and PLT-Redex~\cite{KleinCDEFFMRTF12} that generates parsers,
interpreters, and \LaTeX{} code.  To the best of our knowledge, none of these
tools has been formally verified. As mentioned earlier, Minuska is
directly inspired by \K~framework.

\Cref{fig:impfragment}, which serves as our running example, shows how a
simplified fragment of formal definition of language semantics for a simple imperative
language can look like. (The syntax is explained in \Cref{sec:sugar}, and the full definition of the IMP language semantics
can be found \inextended{in the appendix.)}{in~\cite{MinuskaExt})}

\begin{figure}[t]
\begin{verbatim}
@frames: [simple(X): c[X, STATE]];
@value(X):
    (bool.or(z.is(X), bool.or(bool.is(X),
     term.same_symbol(X, [unitValue[]])))) ;
@strictness:
    [plus of_arity 2 in [0,1], ite of_arity 3 in [0]];
@rule/simple [aexpr.plus]:
    plus[X,Y] => z.plus(X, Y) where bool.and(z.is(X), z.is(Y)) ;
@rule/simple [stmt.ite.true]:
    ite[B, X, Y] => X where bool.eq(B, bool.true()) ;
@rule/simple [stmt.ite.false]:
    ite[B, X, Y] => Y where bool.eq(B, bool.false()) ;
@rule/simple [while.unfold]:
    while[B, S] => ite[B, seq[S, while[B, S]], unitValue[]]
    where bool.true() ;
\end{verbatim}
    \caption{A fragment of the formal semantics of a simple imperative language, written in the sugared concrete syntax of Minuska.}
    \label{fig:impfragment}
\end{figure}

\subsection{Interpreters as testing tools}

In the context of semantic frameworks, interpreters enable language designers
to test that a formally defined language semantics behaves as intended. One
has just to generate an interpreter and test, whether its exhibited behaviour
corresponds to designers' intentions. This has already been observed in the
literature around the \K framework, e.g. in~\cite{Dasgupta0KAR19}.
However, for this testing to be effective, the generated interpreter needs to
\emph{faithfully capture} the language semantics\footnote{Automated testing of a verified compiler for a fixed language
is discussed, e.g., in~\cite{MonniauxGBL23}.}.  Thus, an automatically generated
interpreter formally verified with respect to the semantics would also
increase the trustworthiness of the whole framework, including the languages
modelled in it.

\section{Minuska}\label{sec:architecture}
This section describes Minuska --- a formally verified semantic framework
that generates interpreters 
from formal programming language definitions.  The framework
is implemented in the Coq proof assistant and consists of three main parts:

\begin{enumerate}
\item the language \minusS for static reasoning about program configurations,
\item the language \minusLang for describing computational steps, and
\item a formally verified generator of interpreters.
\end{enumerate}


        

          

While Minuska is directly inspired by the \K~framework and uses similar
concepts, it is based on a different foundational theory -- \emph{first-order term
rewriting over built-in types}. This design decision has been crucial for us
to be able to formally verify Minuska.

The \K~framework uses a complex logic (``matching logic'') with
quantifiers and other connectives to formally capture reasoning about program
configurations and their relations with semantic rules.  If one wanted to
formally verify \K's reasoning in a proof assistant such as Coq, one
would have to formalize not only the matching logic itself but also many
important idioms that \K uses on top of matching logic.  Similarly, one would
need to define a matching logic theory of built-in data types implemented in the \K
framework, such as integers, lists, and maps.  Neither is considered in the
existing work on the mechanization of matching logic in proof assistants, such as
\cite{TusilBH23}.  We also note that the built-in data types of \K
are not implemented in the Metamath-based formalization of matching logic in
\cite{ChenLTR20}.

Since Minuska is based on first-order term rewriting \emph{over built-in
  types}, we do not have to axiomatize the behavior of built-in types -- we
borrow them directly from Coq. The result is much simpler than what would be
needed to formalize \K, and avoids the abovementioned soundness holes.

We do not formally describe the syntax and semantics of MinusLang here and
refer the reader \inextended{to \Cref{sec:formaldefs}.}{to~\cite{MinuskaExt}.} Instead, in the following section,
we informally describe the term rewriting system of Minuska.

\subsection{Term rewriting system}\label{sec:trs}

\paragraph{Configurations}
In Minuska, we represent program configurations as \emph{ground terms} -
first-order terms without variables but possibly with special constants
representing (potentially infinitely many) \emph{built-in values}. These
values inhabit one of the \emph{built-in types}.  For example, a configuration
of a simple imperative program may look like
\begin{verbatim}
c[ (y := x + y; z := y - 1;), (x |-> 3, y |-> 4) ]
\end{verbatim}
where \texttt{c} is a symbol, \texttt{y := x + y; z := y - 1;} is the program to be executed
and \texttt{x |-> 3, y |-> 4} is a built-in value of a dictionary type,
(representing the execution environment) that maps the program variable
\texttt{x} to value $3$ and the program variable \texttt{y} to value
$4$. (Both $3$ and $4$ are values of the built-in integer type.)



\paragraph{Rewriting rules}
At runtime, programs evolve by applying rewriting rules of the form
\texttt{$\alpha$~=>~$\beta$} to configurations. (A set $\Gamma$ of rewriting
rules is called a \emph{rewriting theory}. Such a theory therefore encodes
semantics of a given programming language.)  Both $\alpha$ and $\beta$ are
terms with variables that can \emph{match} program configurations in a given
\emph{valuation} (mapping variables to ground terms). In Minuska, rewriting
rules are conditional, and the \emph{right-hand sides} and side conditions can
use operations defined for our built-in types. (One can think of the built-in
types with their associated operations as forming the \emph{static model} of
Minuska.)

For example, take the following rule \texttt{aexpr.plus}, for binary addition, from the IMP definition in~\Cref{fig:impfragment}
\begin{verbatim}
plus[X,Y] => z.plus(X, Y) where bool.and(z.is(X), z.is(Y))
\end{verbatim}
Its left-hand side matches any term where the ``plus'' symbol is the top level
operator, applied to two operands; the side condition introduced by the
\texttt{where} keyword states that both operands have to be integers; and the
right-hand side adds the two operands using the built-in function for integer
addition.

\subsection{Sugared syntax for MinusLang}
\label{sec:sugar}

However, the rule above is not directly applicable to configurations like the
one in the previous section, which contains both a sequence of statements and
an environment. Even for
simpler configurations of the form \texttt{c[\textit{expression},
  \emph{state}]}  we want to pattern-match only on the
\texttt{\textit{expression}} part of the configuration. This is achieved by
the ``\texttt{@rule/simple [aexpr.plus]:}'' declaration, which states that the
\texttt{aexpr.plus} rule uses the \texttt{simple(CODE): c[CODE, STATE]}
\emph{frame}. In such a frame, only the part corresponding to the
meta-variable \texttt{CODE} is used for pattern matching and subsequent
evaluation (cf. the \texttt{@frames:} declaration in the first line of~\Cref{fig:impfragment}).
Additionally, \texttt{@strictness} declarations like
\begin{verbatim}
    plus of_arity 2 in [0,1]
\end{verbatim}
of \Cref{fig:impfragment}
are used to generate rules that ensure that both arguments of \texttt{plus}
are evaluated to values. A term is considered a value
according to the predicate defined by the keyword \texttt{@value}.
Finally, \texttt{@context} plays a role simillar to that of \texttt{@frames} for
the strictness declarations.

This sugared syntax makes the MinusLang definitions more concise and
readable. For example, using this syntax, a complete definition of IMP programming
language semantics takes less than one page --
\inextended{see \Cref{sec:imp-in-minuska}.}{see~\cite{MinuskaExt}.} The translation from the sugared syntax to pure
MinusLang is defined in \ghref{minuska/theories/frontend.v}{frontend.v}
and \ghref{minuska/theories/default\_everything.v}{default\_everything.v}.

\subsection{Static model}\label{sec:builtins}

In the \K framework, the user can describe
the desired model by using axioms in matching logic,
which can lead to inconsistencies or underspecification,
with no way to detect these.
In Minuska, one has to provide the model directly, which avoids these problems.
The question is then how to describe such a model.
Since Minuska is built inside Coq, one can use the full power of
calculus of inductive construction - the logic of Coq that has been proven
to be sound with respect to set theory.
While it is true that Coq allows its users to assume some axioms without
proving them, this can be easily detected by Coq-provided tools
(or a simple \texttt{grep}).
Moreover, such axioms would not live inside a language definition
but in the metalevel, so a semantic engineer working inside Minuska
has no way of accidentally assuming them.

\paragraph{Default Static Model}
Minuska comes with a default static model.
In the model, built-in values are defined
as if by the following snippet:
\begin{verbatim}
Inductive BuiltInValue :=
| bv_error
| bv_bool (b : bool)
| bv_Z (z : Z)
| bv_sym (s : symbol)
| bv_str (s : string)
| bv_list (m : list GroundTerm)
| bv_finite_dict (m : finite_dict GroundTerm GroundTerm).
\end{verbatim}
That is, a built-in value is either an error value,
a boolean,
an integer, a symbol
(that is used to build terms),
a string, a list of ground terms,
or a finite dictionary from ground terms to ground terms.
Technically, to satisfy Coq's positivity checker,
the last case is implemented as a finite dictionary
from positive natural numbers to ground terms
(we used the \verb|pmap| type that implements
an extensional trie of~\cite{AppelL23} in the Coq stdpp library),
together with the countability property of ground terms
and built-in values.

The provided built-in functions are standard;
these include some arithmetic operations on integers (based on the \texttt{ZArith} module of
the Coq standard library),
boolean operations (based on the \texttt{Bool} module of the Coq standard library),
and operations for manipulating dictionaries (based on the \texttt{pmap} module of stdpp).

\subsection{Interpreter}\label{sec:interpreter}

In the context of Minuska, we consider a \emph{one-step interpreter} to be a partial function
that takes a current configuration (described by a ground term) and performs
one computational step within the current rewriting theory (if possible),
returning the next configuration. As we metioned previously, our
interpreter is implemented wholly within Coq.

Instead of defining a new interpreter for each language, in Minuska we have
defined a simple universal interpreter \textsc{step}. This interpreter is
\emph{universal} in that it takes as an input not only an input
configuration but also a rewriting theory $\Gamma$ (encoding the semantics of
a programming language).
The interpreter works by simply pattern-matching the left-hand side of a rule,
evaluating the side condition in the resulting valuation, and interpreting
the right-hand side using that valuation.  The process of rule selection
is currently somewhat \emph{naive}, based on linear search in the list of rules
representing the rewriting theory $\Gamma$.

\begin{algorithm}
    \caption{Interpreter Step}\label{alg:step}
\begin{algorithmic}

    \Function{naiveSelect}{rewriting theory $\Gamma$, ground term $g$}
        \For{$(l \Rightarrow r\ \texttt{if}\ c) \in \Gamma$}
            \If{\Call{tryMatch}{$l$, $g$} is $\rho \in \mathit{Val}$}
                \If{\Call{evaluteCondition}{$\rho$, $c$}}
                    \State \Return $((l \Rightarrow r\ \texttt{if}\ c), \rho)$
                \EndIf
            \EndIf
        \EndFor
        \State \Return None
    \EndFunction
    \State
    \Function{step}{rewriting theory $\Gamma$, ground term $g$}
        \If{\Call{naiveSelect}{$\Gamma$, $g$} is $((l \Rightarrow r\ \texttt{if}\ c), \rho)$} 
            \State \Return \Call{evaluate}{$\rho$, $r$}
        \Else
            \State \Return None
        \EndIf 
    \EndFunction
\end{algorithmic}
\end{algorithm}

The significant benefit of Minuska is that we can prove that our interpreters are always correct in the
following sense: We say that an interpreter is \emph{sound} with respect to a
rewriting theory iff for every input ground term, if the interpreter returns a
ground term, then the two are related by the semantics of the rewriting
theory.  Conversely, we say that an interpreter is \emph{complete} iff for
every input ground term, the interpreter returns some ground term unless the
input ground term is stuck - i.e., has no successor in the semantics.  Observe
that with such a definition, an interpreter is allowed to resolve
non-determinism arbitrarily.

\begin{theorem}\label{thm:stepCorrect}
    Function \textsc{step} (when applied to a well-formed theory $\Gamma$) is a~sound and complete interpreter (with respect to $\Gamma$).
\end{theorem}
The theorem follows directly from the correctness of its components.
The function \textsc{tryMatch} returns a valuation in which the rule's
left-hand side is satisfied the current program configuration (if there is some);
the function \textsc{evaluateCondition} checks whether the rule's side condition is satisfied in the given valuation;
and the function \textsc{evaluate} returns a concrete program configuration
that satisfies the rule's right-hand side in the given valuation.
A more detailed proof scatch is \inextended{in \Cref{appendix:intepreterCorrectness}}{in~\cite{MinuskaExt}}.

The implementation of the universal interpreter is around 150LoC long, and the
proof of correctness takes about 3kLoC (both excluding supporting
infrastructure).
What made our proofs simpler is that in the implementation,
we abstracted away the details about built-in values:
the interpreter only assumes that built-in values are equipped
with a decidable equality.
The decision procedure for the supported built-in values
is non-trivial ($\sim$1500LoC);
however, because of the abstraction, its details do not complicate
the proof of the interpreter's correctness.


\subsection{Usage and trust base}\label{sec:trust}
The size of trust base of Minuska depends on the way Minuska is used.
With Minuska, one can
\begin{enumerate}
\item put the language definition into a \texttt{*.m}
  file and invoke the \texttt{minuska compile} command
  to check the language definition for well-formedness and generate an
  executable interpreter; or
\item put the language definition into a \texttt{*.m}
  file and use the \texttt{minuska def2coq} command to generate a Coq (\texttt{*.v}) file containing the language definition and
the associated interpreter, then load the generated file into Coq and use it from within Coq; or
\item write the language definition directly in Coq.
\end{enumerate}
In the first situation one needs to trust the hand-written OCaml code that implements
the \texttt{minuska} executable, as well as Coq's extraction mechanism and OCaml compiler;
in the second case, the trust base consists of the \texttt{minuska} executable
and Coq's kernel; in the last case, only the kernel of Coq needs to be trusted.
In principle, the trust base of the first case could be reduced by using
a formally verified extraction mechanism (to OCaml~\cite{CoqOcamlVerifiedExtraction} or C~\cite{Anand2016CertiCoqA})
and a formally verified compiler
(such as CompCert~\cite{BoldoJLM13}).

\subsection{Concrete syntax and parsing}
While Minuska internally operates on the abstract syntax of the given program,
and also the semantic rules of a programming language are written with respect to the language's abstract syntax,
Minuska provides an interface to connect a user-provided parser for the target language.



\section{Evaluation}\label{sec:evaluation}

We evaluate Minuska from two perspectives,
\emph{expressivity} and \emph{performance},
with respect to the \K-framework implementation of \cite{ChenLTR20},
as well as on its own.
For performance evaluation with respect to \cite{ChenLTR20},
we have to choose benchmarks supported by both Minuska and \K.
In \cite{ChenLTR20}, the authors choose two sets of benchmarks:
a rewriting system (named ``two-counters'') implementing summation from $1$ to $n$ of unary-encoded natural numbers,
executed with various inputs, and a selection of benchmarks from the
rewriting engines competition (REC \cite{RECPaper}), with reduced inputs, adapted to \K.
However, the adapted REC benchmarks use
user-defined function symbols, which (as discussed in \Cref{sec:intro})
allows for a logical inconsistency in a \K-based language definition.
Moreover, the REC benchmarks were designed for benchmarking rewriting
engines, and not semantic frameworks,
and thus exercise neither the idioms used in developments of language definitions
using \K (such as \emph{strictness} for specifying evaluation order) nor the supporting infrastructure
for these as implemented in both \K (\emph{strictness attributes})
and Minuska (\emph{strictness declaration}).

\begin{figure}[h!]
\begin{verbatim}
@rule [step]: 
state[M, N] => state[z.minus(M, [(@builtin-int 1)]), z.plus(N, M)]
where z.lt([(@builtin-int 0)],M);
\end{verbatim}
        
\begin{minipage}[t]{.45\textwidth}
\hrulei
\begin{tabular}{r|c}
    count-to & Native compute time [s] \\ \hline
       1'000 & 0.01 \\
      10'000 & 0.06 \\
     100'000 & 0.45 \\
   1'000'000 & 4.39 \\
\end{tabular}    
\end{minipage}\hspace{0.02\linewidth}
\begin{minipage}[t]{.45\textwidth}
\hrulei
\begin{tabular}{r|cc}
count-to & Compute time [s] & coqc time [s] \\ \hline
   10  & 0.028 & 0.75 \\
   100  & 0.044 & 0.78 \\
 1'000  & 0.213 & 0.92 \\
10'000  & 1.871 & 2.58
\end{tabular}
\end{minipage} %
\caption{``two-counters'': the running example of \cite{ChenLTR20}
in Minuska. ``Native compute time'' means wall-clock time of the system executed in an extracted and compiled standalone interpreter.
``Compute time'' is the time needed to evaluate the benchmark using Coq's \texttt{Compute} command,
and ``coqc time'' is time needed to evaluate a Coq (\texttt{*.v}) file containing the command.}
\label{fig:twocounters}
\end{figure}

On the other hand, the ``two-counters'' benchmark is too simple for our purposes
(\Cref{fig:twocounters}).  Therefore, to compare Minuska with respect to \K,
we crafted two new benchmarks that are simple enough to be handled by \K and
simultaneously idiomatic enough in the sense that they use strictness
idioms. As mentioned in \cite{ChenLTR20}, their tool does not support the
native strictness attribute of \K. Therefore, in the \K  version of the benchmarks, we
emulated this feature manually in the same style as done in the version of the
benchmarks for Minuska.  Since \K does not support generating proofs for (its)
builtin integers, we also restrain from using builtin values, with the
exception of builtin booleans which are used in the side conditions of rules
both in \K and in Minuska.  The two new benchmarks are: computing the $n$-th
member of the Fibonacci sequence, and computing the factorial of $n$, both
done in the iterative manner over unary-encoded natural numbers.

\subsection{Performance of Minuska vs \K}

\begin{figure}
    \begin{tabular}{c|ccc|cc}
    name & \K parameters [s] & \K cert-gen [s] & \K cert-check [s] & M compute [s] & M coqc [s] \\ \hline
    tc 10    & 38.01          & 1.40      & 2.71         & 0.03 & 0.75    \\
    tc 20    & 79.14          & 2.41      & 2.97         & 0.03 & 0.75   \\
    tc 50    & 192.64         & 6.92      & 2.84         & 0.04 & 0.76   \\
    tc 100   & 377.83         & 14.44     & 2.81         & 0.04 & 0.78    \\
    fib 5    & 95.28          & 4.92      & 2.83         & 0.20 & 0.89 \\
    fact 3   & 247.50         & 19.71     & 3.35         & 0.48 & 1.40 \\
    \end{tabular}
    \caption{Running times of shared \K/Minuska benchmarks.
    The meaning of the columns is as follows.
    For \K,
    ``\K parameters'' is the time required to generate proof parameters for a single execution,
    including time required argument parsing;
    ``\K cert-gen'' is the time of certificate generation from those parameters;
    and ``\K cert-check'' is the time required for checking those certificates.
    For Minuska, ``M compute'' is the time to evaluate the benchmark using Coq's Compute mechanism,
    measured inside Coq using its Time command;
    and ``M coqc'' is the the time required for Coq to process a file with the benchmark,
    as measured by the GNU \texttt{time} utility.
    }
    \label{fig:benchmarksshared}
    \end{figure}

To benchmark \K, we use the infrastructure of \cite{ChenLTR20},
which we modified so as to measure also the time to
execute a given program on a given input while generating
so-called ``proof parameters''.
We measure the performance on Intel Core I7 10th Gen
(model name Intel(R) Core(TM) i7-10510U CPU @ 1.80GHz)
in a virtual machine limited to 12 GiB of RAM,
running Ubuntu 20.04 (for which the available artifact is intended).
To benchmark Minuska, we use the same hardware running virtualized
Ubuntu 23.10, with the same memory limit.
We use a newer OS for Minuska because it depends on Coq 8.19, which was not available in Ubuntu 20.04.

Our benchmarking results are in \Cref{fig:benchmarksshared};
the ``two-counters'' benchmark is abbreviated as ``tc''.
We argue that the most interesting measure of \K's
interpreter is the \emph{sum} of all three measured times:
to ensure that one gets a valid result, one needs to first run the interpreter
while generating proof parameters, then to generate a certificate,
and finally to check it. Unforunately, the times required for generating
proof parameters are so high because of the rather naive approach
applied: the \K tool is called repeatedly as many times as there are
computational steps needed for the execution.
We note that the ``\K parameters'' column is not reported in~\cite{ChenLTR20}.

Regarding Minuska, the difference between the 'Compute' time and the 'coqc' time
is almost constant, over $0.6$ seconds; we intepret this difference
as the time required to load Coq and the Minuska library.
Minuska's certified (but naive) interpreter outperforms
\K's certifying interpreter by a large margin.
Perhaps surprisingly, it is faster to start Coq and perform
a verified computation inside than to verify the proof certificate
generated by \K for an equivalent computation.

\subsubsection{Discussion}
In our opinion, one reason why Minuska's verified
but naive interpreter outperforms \K's certifying interpreter
is the size of the generated proof objects.
This is likely due to the choice of \emph{matching logic}
as the logic of the proof certificates:
matching logic is rather low-level and its Hilbert-style proof system
makes it hard to come up with compact proofs without resorting to metalevel reasoning.
Even a very fast proof checker can hardly compensate for the size of \K's proof certificates,
which is often in hundreds of megabytes of text data.

\subsection{IMP semantics  in Minuska}

In addition, we use the full definition of a simple imperative language IMP,
adapted from \cite{ChenLTR20},
to demonstrate that Minuska can do more than the version of \K of \cite{ChenLTR20}:
to interpret a program in a programming language,
and to do so with a reasonable performance.
According to \Cref{fig:bench-imp-count2n}, one iteration of the main loop
of \texttt{imp-count-to} takes about half a second when interpreted by Coq,
and about 10 milliseconds when executing natively.
The example uses built-in integers, as well as built-in
maps; the source code is available in \ghref{languages/imp/imp.m}{imp.m}
and \inextended{in~\Cref{fig:imp-in-minuska}}{in~\cite{MinuskaExt}}.

\begin{figure}
\begin{minipage}[t]{.45\textwidth}
\hrulei    
\vspace{4pt}
\begin{verbatim}
/* imp-count-to-n */
n := $arg ;
sum := 0 ;
while (1 <= n) {
    sum := sum + n ;
    n := n + (-1) ;
};
sum
\end{verbatim}
\end{minipage}\hspace{0.02\linewidth}
\begin{minipage}[t]{.45\textwidth}
\hrulei
    \begin{tabular}{c|ccc}
        \$arg & Compute [s] & coqc [s] & native \\ \hline
        1  & 0.66 & 1.49 & 0.02 \\
        2  & 1.15 & 1.86 & 0.02 \\
        3  & 1.57 & 2.28 & 0.03 \\
        4  & 2.03 & 2.78 & 0.04 \\
        5  & 2.45 & 3.13 & 0.04 \\
        6  & 2.95 & 3.62 & 0.05 \\
        7  & 3.41 & 4.16 & 0.06 \\
    \end{tabular}
\end{minipage}
    \caption{Running times of a program for computing $\Sigma_{i=i}^{n} i$
        in a simple imperative language IMP with semantics defined in Minuska.
        'Compute' means the time of the 'Compute' command within Coq,
        'coqc' means the total time of the `coqc' command on a file containing the Compute command,
        and `native' means the time of the native code generated by the OCaml compiler for the interpreter extracted
        from Coq into OCaml.
    }
    \label{fig:bench-imp-count2n}
\end{figure}

%

\subsection{Other benchmarks}

To get a fuller understanding of Minuska's performance, we have used an additional
set of benchmarks based on computing Fibonacci numbers and factorials. The results are in \Cref{fig:benchminuska}.
Here ``unary-fib'' is the same Fibonacci system
as in \Cref{fig:benchmarksshared},
and ``unary-fact'' is the same code for computing factorials,
both measured with growing input parameter.
The ``unary'' prefix means that a naive unary encoding of natural numbers
is used, with addition of two numbers being computed in
a linear number of steps in the value of the first addend.
Also, in both cases a naive recursive algorithm is used.
The ``native-fib'' collection of benchmarks exercise a different implementation
of the Fibonacci sequence: the code is more imperative in style,
always remembering the last two values of the sequence
instead of recursing to compute them; moreover, it uses Minuska's
built-in integers implemented using Coq's binary integers.

\begin{figure}
        \begin{center}
            \begin{minipage}[t]{.45\textwidth}
            
        \begin{tabular}{c|cc}
            benchmark & Compute [s] & coqc [s] \\ \hline
            native-fib 1  & 0.026 & 0.74 \\
            native-fib 11  & 0.042 & 0.80 \\\
            native-fib 21  & 0.058 & 0.79 \\
            native-fib 31  & 0.068 & 0.78 \\
            native-fib 41  & 0.079 & 0.79 \\
            native-fib 51  & 0.11 & 0.79 \\
            native-fib 61  & 0.125 & 0.81 \\
            native-fib 71  & 0.141 & 0.82 \\
            native-fib 81  & 0.152 & 0.82 \\
            native-fib 91  & 0.161 & 0.84 \\
            native-fib 101  & 0.172 & 0.86 \\
            \end{tabular}
        
            \end{minipage}
            \hspace{0.02\linewidth}
            \begin{minipage}[t]{.45\textwidth}
            
\begin{tabular}{c|cc}
    benchmark & Compute [s] & coqc [s] \\ \hline
    unary-fib 1  & 0.094 & 0.77 \\
    unary-fib 2  & 0.102 & 0.78 \\
    unary-fib 3  & 0.12 & 0.80 \\
    unary-fib 4  & 0.151 & 0.84 \\
    unary-fib 5  & 0.199 & 0.89 \\
    unary-fib 6  & 0.285 & 0.98 \\
    unary-fib 7  & 0.428 & 1.21 \\
    unary-fib 8  & 0.818 & 1.53 \\
    unary-fib 9  & 1.233 & 2.02 \\
    unary-fib 10  & 2.171 & 2.91 \\
    unary-fib 11  & 3.313 & 4.04 \\
    unary-fact 1  & 0.159 & 0.90 \\
    unary-fact 2  & 0.222 & 0.95 \\
    unary-fact 3  & 0.476 & 1.40 \\
    unary-fact 4  & 1.107 & 1.85 \\
    unary-fact 5  & 3.108 & 3.86 \\
    unary-fact 6  & 12.785 & 13.60 \\
    \end{tabular}
        \end{minipage}
    \end{center}
    \caption{Other benchmarks of Minuska}
    \label{fig:benchminuska}
\end{figure}

\subsection{Development effort}

The part of Minuska reported in this paper
cost about 4 person-months to develop
and consists of approximately 16kLOC of Coq code,
only approximately 700LOC of which are specifications (of MinusLang etc.), the
rest being implementation and proofs.


\section{Conclusion}

We have introduced Minuska: a formally-verified programming language framework
capable of generating interpreters from language definitions.
The generated interpreters, while using a naive pattern matching algorithm,
have better performance than the certifying interpreters of~\cite{ChenLTR20},
while also having the advantage of being formally verified to behave correctly with respect
to the given language definition.
However, more future work remains.

\paragraph{Future developments}
When it comes to semantic frameworks, having a verified semantic-based interpreter is a good start,
but one can go further.
\begin{enumerate}
\item As has been demonstrated by \cite{StefanescuPYLR16}, 
semantic frameworks can implement symbolic execution and deductive verification facilities.
\item Although our pattern-matching-based interpreter is more performant than the certifying competition,
it is still rather slow. We believe we could reuse some existing work for performant pattern matching
that has been done in the context of compilers - for example, \cite{CompilingPM}.
The \K tool already has an unverified implementation of some of these ideas.
\item In both Minuska and \K,
the constructs for \emph{evaluation contexts}, \emph{strictness}, and \emph{order of evaluation}
of subexpression have no formal semantics, but are implemented only as a syntactic sugar
on top of ordinary rewriting rules. It would be better if these had a formal semantics,
since it would make reasoning about a language semantics more high-level.
\item It would also be interesting to model some of the larger languages defined in \K in Minuska as well.
\end{enumerate}

\section{Acknowledgements}
    We are grateful to Traian Florin Șerbănuță for his input to discussions around Minuska,
    and to anonymous reviewers for their feedback on a previous version of this paper.    

\bibliographystyle{splncs04}
\inextended{
  \bibliography{bibliography}
}{
  \bibliography{bibliography,extendedbibliography}
}
\inextended{\appendix\newpage
\section{Formal definitions}\label{sec:formaldefs}

\subsection{Language \minusS}
For reasoning about static program configurations, Minuska uses
a language named \minusS.
The key concepts here are \emph{ground terms} (representing program configurations),
\emph{symbolic terms}
(representing left sides of rewriting rules), 
\emph{expression terms} (representing right sides of rewriting rules),
\emph{constraints}, and
\emph{satisfaction relations} between these.

\subsubsection{Syntax}
Let us fix a type $\mathit{Sym}$ of \emph{symbols},
a type $\mathit{Var}$ of \emph{variables},
a type $\mathit{B}$ of \emph{builtin values};
a type $\mathit{F}$ of \emph{builtin function names};
and a function $\mathit{ar} : F \to \mathbb{N}_0$.

Given a type $X$, the type $\mathcal{T}(X)$ of \emph{terms over $X$} is defined inductively to be inhabited by
\begin{itemize}
    \item $x$, where $x \in X$; and
    \item $s [t_n,\ldots,t_n]$, where $s \in \mathit{Sym}$
          and $t_i \in \mathcal{T}(X)$ (for every $i \in \{ 1, \ldots, n \}$).
\end{itemize}
Then, the type $\mathcal{T}_g$ of \emph{ground terms} is defined to be exactly $\mathcal{T}(\mathit{B})$;
the type $\mathcal{T}_s$ of \emph{symbolic terms} is defined to be exactly $\mathcal{T}(\mathit{B} + \mathit{Var})$,
where $\mathit{B} + \mathit{Var}$ represents the sum type of $\mathit{B}$ and $\mathit{Var}$;
and the type $\mathcal{T}_e$ of \emph{expression terms} is defined to be exactly $\mathcal{T}(E)$,
where $\mathit{E}$ is the set of \emph{expressions} defined inductively as follows:
\begin{itemize}
    \item a ground term $g \in \mathcal{T}_g$ is an expression;
    \item a variable $v \in \mathit{Var}$ is an expression;
    \item $f(e_1,\ldots,e_n)$, where $f \in \mathit{F}$ and $e_i$ is an expression (for every $i \in \{ 1, \ldots, n \}$),
          is an expression. 
\end{itemize}
We denote $\mathit{CS} = E \times E$ to be the set of \emph{side conditions}.
For any symbol $s$, we take the freedom to write $s$ instead of $s[]$
in any context requiring (ground, symbolic, or expression) terms.
We also define the substitution function $\mathit{subst} : \mathcal{T}_s \times \mathit{Var} \times \mathcal{T}_s \to \mathcal{T}_s$
in the usual way - that is, inductively, by
\begin{itemize}
    \item $\mathit{subst}(x, x, t) = t$;
    \item $\mathit{subst}(y, x, t) = y$ if $y \in \mathit{Var}$ and $x \neq y$; and
    \item $\mathit{subst}(s[t_1,\ldots,t_n], x, t) = s[\mathit{subst}(t_1,x,t),\ldots,\mathit{subst}(t_n,x,t)]$,
\end{itemize}
and use the notation $t[t^\prime/x]$ for $\mathit{subst}(t, x, t^\prime)$.
We let $\mathit{FV}(t)$ to denote the set of (free) variables of a symbolic term,
or an expression, or a set or list thereof.

\subsubsection{Semantics}

A \emph{static model} $\Sigma$ consists of a total function
$\Sigma_f : (\mathcal{T}_g)^{ar(f)} \to \mathcal{T}_g$
that is the interpretation of the function name $f \in F$.
Let us fix some static model $\Sigma$. Let $\mathit{Val}$ denote the type of \emph{valuations}
$\rho : \mathit{Var} \rightharpoonup \mathcal{T}_g$ - that is, the type of partial functions from variables to ground terms.
We define a relation $\vDash_E \subseteq \mathit{Val} \times \mathcal{T}_g \times \mathit{E}$,
representing satisfaction $\rho,g \vDash_E e$ between ground terms and expressions,
inductively:
\begin{itemize}
\item $\rho,g \vDash_E g$ for any ground term $g$;
\item $\rho,g \vDash_E v$ for any variable $v$ and ground term $g$ such that $\rho(v) = g$; and
\item $\rho,g \vDash_E f(e_1,\ldots,e_n)$ if there are some ground terms $g_1,\ldots,g_n$ such that 
    $\rho,g_i \vDash_E e_i$ (for every $i \in \{ 1,\ldots, n \}$) and $\Sigma_f(g_1,\ldots,g_n) = g$.
\end{itemize}
Furthermore, we define the relation $\vDash_O \subseteq \mathit{Val} \times \mathcal{T}_g \times (\mathit{B} + \mathit{Var})$ as follows:
\begin{itemize}
\item $\rho,b \vDash_O b$ if $b$ is a builtin value; and
\item $\rho,g \vDash_O v$ if $v$ is a variable such that $\rho(v) = g$.
\end{itemize}

Now we can define the relation $\vDash_s \subseteq \mathit{Val} \times \mathcal{T}_g \times \mathcal{T}_s$ inductively:
\begin{itemize}
\item $\rho,g \vDash_s \mathit{bv}$ if $\rho,g \vDash_O \mathit{bv}$ for any $\mathit{bv} : \mathit{B} + \mathit{Var}$;
\item $\rho,s[g_1,\ldots,g_2] \vDash_s s[t_1,\ldots,t_n]$
    where $g_i : \mathcal{T}_g$ and $\rho,g_i \vDash_s t_i$
    (for every $i \in \{ 1, \ldots, n \}$).
\end{itemize}
The relation $\vDash_e \subseteq \mathit{Val} \times \mathcal{T}_g \times \mathcal{T}_e$ is defined
by a similar schema with different base case:
\begin{itemize}
\item $\rho,g \vDash_e e$ if $\rho,g \vDash_E e$ for any $e : \mathit{E}$;
\item $\rho,s[g_1,\ldots,g_2] \vDash_s s[t_1,\ldots,t_n]$
    where $g_i : \mathcal{T}_e$ and $\rho,g_i \vDash_s t_i$
    (for every $i \in \{ 1, \ldots, n \}$).
\end{itemize}
A side condition $(e_1, e_2)$ \emph{holds} in valuation $\rho$, written $\rho \vDash e_1 = e_2$, if for any $t_1,t_2 \in \mathcal{T}_g$,
if $\rho,t_1 \vDash e_1$ and $\rho,t_2 \vDash e_2$ then $t_1 = t_2$.
A list $\mathit{cs}$ of side conditions \emph{holds} in valuation $\rho$, written $\rho \vDash \mathit{cs}$, if $\rho \vDash c$ for every $c \in \mathit{cs}$.
When the types of the operands are clear from the context, we take the freedom to omit the subscripts
and write simply $\vDash$ for any of $\vDash_e$, $\vDash_s$, $\vDash_E$ and $\vDash_O$.

\subsection{Language \minusLang}

Language \minusLang builds on top of \minusS and adds
\emph{actions}, \emph{rewriting rules},
and \emph{loading} and \emph{unloading} functions.
As an example, let us consider the rewriting rule
\begin{equation*}
    \ruleKeyword \texttt{plus} [ X, Y ] \Rightarrow \mathit{plusZ}(X, Y) \, .
\end{equation*}
which ensures, for example, that the ground terms $\texttt{plus}[3,4]$ and $7$
are in the one-step rewriting relation: $\texttt{plus}[3,4] \rightsquigarrow 7$.

\subsubsection{Syntax}
Let us fix
a type $\mathit{Act}$ of \emph{actions}.
We say that a $w \in \mathit{Act}^*$ is an \emph{action word}.
A \emph{rewriting rule} is a quadruple $(l, r, \mathit{cs}, a)$,
where $l \in \mathcal{T}_s$, $r \in \mathcal{T}_e$, $\mathit{cs}$ is a list of side conditions
of the shape $(e_1, e_2)$ where $e_1,e_2 \in \mathit{E}$, and $a \in \mathit{Act}$,
such that $\mathit{FV}(r) \cup \mathit{FV}(\mathit{cs}) \subseteq \mathit{FV}(l)$;
and a \emph{rewriting theory} is a finite set of rewriting rules.

\subsubsection{Semantics}

Given a rewriting theory $\Gamma$, let $\minusRelPrime{\Gamma}{\_} \subseteq \mathcal{T}_g \times \mathit{Act} \times \mathcal{T}_g$
be the ternary relation between ground terms and actions defined by:
$g_1 \minusRelPrime{\Gamma}{a} g_2$ iff there exists some rule $(l, r, \mathit{cs}, a) \in \Gamma$ and a valuation $\rho$
such that $\rho,g_1 \vDash l$, $\rho,g_2 \vDash r$, and $\rho \vDash \mathit{cs}$.
We extend $\minusRelPrime{\Gamma}{\_}$ to action words, by defining 
$g_1 \minusRelPrime{\Gamma}{a_1 \cdot...\cdot a_n} g_{n+1}$ to hold iff there exist some
$g_2,\ldots,g_{n} \in \mathcal{T}_g$ such that $g_i \minusRelPrime{\Gamma}{a_i} g_{i+1}$ (for $i \in \{ 1, \ldots, n \} $).

The above definitions are mechanized in Coq,
in \ghref{minuska/theories/spec.v}{spec.v}.
The well-formedness condition on variables of rewriting rule
is mechanized \ghref{minuska/theories/spec_interpreter.v\#L64-L82}{separately}.
The interpreter and its correctness theorem are mechanized in\ghref{minuska/theories/naive\_interpreter.v}{naive\_interpreter.v}.

\subsection{Interpreter correctness}\label{appendix:intepreterCorrectness}
\begin{algorithm}
  \caption{TryMatch}\label{alg:trymatch}
\begin{algorithmic}
  \Function{tryMatch}{pattern $\varphi$, ground term $g$}
      \If{$\varphi$ is variable $x$}
        \State \Return $\{ (x, \varphi) \}$
      \EndIf
      \If{$\varphi$ is ground term $g^\prime$}
        \If{$g^\prime = g$}
          \State \Return $\emptyset$
        \Else
          \State \Return None
        \EndIf
      \EndIf
      \If{$g$ is a builtin value}
        \State \Return None
      \EndIf
      \State $s^\prime(\varphi_1,...,\varphi_n) \gets \varphi$
      \State $s(t_1,\ldots,t_n) \gets g$
      \If{$s \not = s^\prime$}
        \State \Return None
      \EndIf
      \State $v_1,\ldots,v_n \gets \Call{tryMatch}{\varphi_1, t_1}, \ldots, \Call{tryMatch}{\varphi_n, t_n}$
      \If{Some $v_i$ is None}
        \State \Return None
      \EndIf
      \If{$v_i(x)$ and $v_j(x)$ are defined but distinct for some $x$ and $i \not = j$}
        \State \Return None
      \EndIf
      \State \Return $v_1 \cup \ldots \cup v_n$
  \EndFunction
\end{algorithmic}
\end{algorithm}

\begin{lemma}[\textsc{tryMatch} correct]\label{lem:tryMatch}
  For any ground term $g$, any symbolic term $\varphi$, and any valuation $\rho$,
  \begin{enumerate}
    \item if $\textsc{tryMatch}(g, \varphi)$ returns $\rho$, then $\rho, g \vDash_s \varphi$; and
    \item if $\rho, g \vDash_s \varphi$, then $\textsc{tryMatch}(g, \varphi)$ returns some valuation $\rho^\prime$
      such that
      \begin{itemize}
        \item $\rho^\prime$ is defined exactly on variables of $\varphi$; and
        \item $\rho$ extends $\rho^\prime$.
      \end{itemize}
  \end{enumerate}
  We call the first property \emph{soundness} and the second one \emph{completeness}.
\end{lemma}
\begin{proof}
  Both properties are proved by straightforward but technical induction on the size of $\varphi$
  and case analysis of both $\varphi$ and $g$.
\end{proof}

\begin{lemma}[\textsc{evaluateCondition}]\label{lem:evaluateCondition}
  For any side condition $c$, $\textsc{evaluateCondition}(\rho, c)$ returns True if $\rho \vDash c$
  and False otherwise.
\end{lemma}
\begin{proof}
It is easy to write \textsc{evaluateCondition} such that the lemma holds by induction on size of the condition.
\end{proof}

\begin{lemma}[\textsc{evaluate} correct]\label{lem:evaluateExpression}
  For any ground term $g$, an expression term $e$, and any valuation $\rho$,
  $\textsc{evaluate}(\rho, e)$ returns $g$ if and only if $\rho, g \vDash_e e$.
\end{lemma}
\begin{proof}
  It is easy to write \textsc{evaluate} such that the lemma holds by induction on size of the expression.
\end{proof}

\begin{proof}(of \Cref{thm:stepCorrect})
  Specifically, for soundness, assume $\textsc{step}(\Gamma, g)$ returns $g^\prime$.
  That can happen only if $\textsc{naiveSelect}(\Gamma, g)$ returns
  some $((l \Rightarrow r\ \texttt{if}\ c), \rho)$
  such that $\textsc{Evaluate}(\rho, r)$ returns $g^\prime$.
  Because $\textsc{naiveSelect}$ is performing only a simple linear search in $\Gamma$,
  it follows that $(l \Rightarrow r\ \texttt{if}\ c) \in \Gamma$, $\textsc{tryMatch}(l, g)$
  had to return $\rho$, $\textsc{evaluate}(\rho, r)$ had to return $g^\prime$, and $\textsc{evalauteCondition}(\rho, c)$ had to return True.
  But by correctness of these components, we have $\rho, l \vDash g$, $\rho, r \vDash g^\prime$,
  and $\rho \vDash c$. But these are exactly the conditions relating $g$ and $g^\prime$ in the rewriting relation.

  Completeness is almost the reverse of the above argument, with two caveats.
  \begin{enumerate}
    \item One has to take care of the fact that the semantic relation between
    $g$ and $g^\prime$ gives us only \emph{some} valuation, not necessarily a minimal one (that would be returned by \textsc{tryMatch}).
    That is why soundness and completenss of \textsc{tryMatch} are not symmetrical.
    \item Another catch is that the interpreter returns only \emph{some} successor of $g$, not necessarily $g^\prime$, because the language
    definition might be nondeterministic. That is the reason why the definition of completeness of the interpreter is formulated
    differently than its soundness.
  \end{enumerate}

\end{proof}

\newpage
\section{Language definition of IMP in Minuska}\label{sec:imp-in-minuska}

Here, in \Cref{fig:imp-in-minuska}, we present a full language definition of a simple imperative programming language, IMP, in 
the concrete syntax of Minuska, as accepted by the \texttt{minuska} script.
\begin{figure}[h]
\begin{verbatim}
@frames: [simple(CODE): c[builtin.cseq [CODE,REST], STATE]];
@value(X): (bool.or(z.is(X), bool.or(bool.is(X),term.same_symbol(X, [unitValue[]])))) ;
@context(HOLE): c[HOLE, STATE];
@strictness: [ neg of_arity 1 in [0],
  plus of_arity 2 in [0,1], minus of_arity 2 in [0,1],
  assign of_arity 2 in [1], seq of_arity 2 in [0], ite of_arity 3 in [0],
  eq of_arity 2 in [0,1], le of_arity 2 in [0,1], lt of_arity 2 in [0,1]
];
@rule [init]: builtin.init[X]
  => c[builtin.cseq[X, builtin.empty_cseq[]], map.empty()] where bool.true();
@rule/simple [aexpr.plus]: plus[X,Y] => z.plus(X, Y) where bool.and(z.is(X), z.is(Y)) ;
@rule/simple [aexpr.minus]: minus[X,Y] => z.minus(X, Y) where bool.and(z.is(X), z.is(Y)) ;
@rule [var.assign]: c[builtin.cseq[assign[X,V],REST], STATE]
=> c[builtin.cseq[unitValue[], REST], map.update(STATE, X, V)]
    where bool.and(term.same_symbol(X, [var[]]), z.is(V)) ;
@rule [var.lookup]: c[builtin.cseq[X, REST], STATE]
=> c[builtin.cseq[map.lookup(STATE, X), REST], STATE] where term.same_symbol(X, [var[]]) ;
@rule/simple [stmt.seq]: seq[unitValue[], X] => X where bool.true() ;
@rule/simple [bexpr.eq]: eq[X, Y] => z.eq(X, Y) where bool.and(z.is(X), z.is(Y)) ;
@rule/simple [bexpr.le]: le[X, Y] => z.le(X, Y) where bool.and(z.is(X), z.is(Y)) ;
@rule/simple [bexpr.lt]: lt[X, Y] => z.lt(X, Y) where bool.and(z.is(X), z.is(Y)) ;
@rule/simple [bexpr.neg]: not[X] => bool.neg(X) where bool.is(X) ;
@rule/simple [stmt.ite.true]:
  ite[B, X, Y] => X where bool.eq(B, bool.true()) ;
@rule/simple [stmt.ite.false]:
  ite[B, X, Y] => Y where bool.eq(B, bool.false()) ;
@rule/simple [while.unfold]:
  while[B, S] => ite[B, seq[S, while[B, S]], unitValue[]] where bool.true() ;
\end{verbatim}
\caption{A language definition of IMP in Minuska}
\label{fig:imp-in-minuska}
\end{figure}
}{}

\end{document}